\documentclass[runningheads,letterpaper,10pt]{llncs}

\newcommand{\longversion}[1]{#1}
\newcommand{\shortversion}[1]{}

\usepackage{booktabs}
\usepackage[bookmarks=false,pdfpagelabels=false]{hyperref}
\usepackage{amsfonts,amsmath} 
\usepackage{xcolor} 
\usepackage{tikz}
\usepackage{tikz-qtree}
\usepackage{todonotes}

\hypersetup{breaklinks,pdfdisplaydoctitle,colorlinks,linkcolor=red!50!black,citecolor=green!50!black}

\usepackage{mathrsfs}
\usepackage[linesnumbered,boxed]{algorithm2e}

\longversion{}
\shortversion{}

\shortversion{\include{spacesaver}}

\begin{document}

\title{Faster algorithms to enumerate hypergraph transversals}

\author{Manfred Cochefert\inst{1}
\and Jean-Fran\c{c}ois Couturier\inst{2} 
\and Serge Gaspers\inst{3} \inst{4}
\and \\Dieter Kratsch\inst{1}%
}

\authorrunning{M. Cochefert, J.-F. Couturier, S. Gaspers, and D. Kratsch}

\institute{ %
LITA, Universit\'e de Lorraine, Metz, France.
\email{manfred.cochefert@gmail.com}, 
\email{dieter.kratsch@univ-lorraine.fr}
\and
CReSTIC, Universit\'e de Reims, France.
\email{jean-francois.couturier@univ-reims.fr}
\and
University of New South Wales, Sydney, Australia.
\email{sergeg@cse.unsw.edu.au}
\and
NICTA, Sydney, Australia
}

\maketitle

\begin{abstract}
A transversal of a hypergraph is a set of vertices intersecting each hyperedge.
We design and analyze new exponential-time algorithms to enumerate all 
inclusion-minimal transversals of a hypergraph. 
For each fixed $k\ge 3$, our algorithms for hypergraphs of rank $k$, where the rank is the maximum size of a hyperedge, outperform the previous best.
This also implies improved upper bounds on the maximum number of minimal transversals in 
$n$-vertex hypergraphs of rank $k\ge 3$.
Our main algorithm is a branching algorithm whose running time is analyzed with Measure and Conquer.
It enumerates all minimal transversals of hypergraphs of rank $3$ 
in time $O(1.6755^n)$. 
Our algorithm for hypergraphs of rank $4$ is based on iterative compression.
Our enumeration algorithms improve upon the best known algorithms 
for counting minimum transversals 
in hypergraphs of rank $k$ for $k\ge 3$ and for computing a minimum 
transversal in hypergraphs of rank $k$ for $k\ge 6$. 
%
\end{abstract}

\section{Introduction}  \label{s:intro}

A \emph{hypergraph} $H$ is a couple $(V,E)$, where $V$ is a set of vertices and $E$ is a set of subsets of $V$\longversion{; these subsets are} called \emph{hyperedges}. A \emph{transversal} of $H$ is a subset of vertices $S\subseteq V$ such that each hyperedge of $H$ contains at least one vertex from $S$. A transversal of $H$ is \emph{minimal} if it does not contain a transversal of $H$ as a proper subset. The \emph{rank} of a hypergraph $H$ is the maximum size of a hyperedge.
Finding, counting and enumerating (minimal) transversals\longversion{, respectively (minimal) hitting sets of set systems} fulfilling certain constraints are fundamental problems 
in Theoretical Computer Science with many important applications, for example in artificial intelligence, biology, logics, relational and distributed databases, Boolean switching theory and model-based  diagnosis. Various of those applications are described 
in Section~6 of~\cite{EiterG03}.
The notions \emph{hitting set} and transversal are synonymous, both of them name a subset of 
elements (vertices) having non empty intersection with each subset (hyperedge) of a given set system (hypergraph). We shall use both notions interchangeably usually speaking of transversals 
in hypergraphs whenever adressing enumeration, while 
speaking of hitting sets
when adressing optimization and counting.  
The \textsc{Minimum Hitting Set} problem
is  a well-studied NP-complete problem that, like its dual \textsc{Minimum Set Cover}, belongs
to the list of 21 problems shown to be NP-complete by Karp in 1972~\cite{Karp72}. It can be seen as an  extension of the fundamental graph problems \textsc{Minimum Dominating Set}, 
\textsc{Minimum Vertex Cover} and \textsc{Maximum Independent Set}; all of them 
also 
belonging to Karp's list~\cite{Karp72}. 
These fundamental NP-complete 
problems have been studied extensively from many algorithmic views; among them
exact, approximation and parameterized algorithms.

{\bf Prior to our work.}
Wahlstr\"om studied \textsc{Minimum Hitting Set}
on hypergraphs of rank $3$ and achieved an $O(1.6359^n)$ time and polynomial space
algorithm 
as well as an $O(1.6278^n)$ time and exponential space algorithm~\cite{Wahlstrom04,Wahlstrom07}.
Here, $n$ denotes the number of vertices.
In an attempt to show that iterative compression
can be useful in exact exponential-time algorithms, Fomin et al. studied 
\textsc{Minimum Hitting Set} on hypergraphs of rank at most $k$ for any fixed $k\ge 2$
as well as the problem of counting minimum hitting sets and achieved best known running times
for most of these problems without improving upon Wahlstr\"om's algorithms~\cite{FominGKLS10tcs}.
(See also Table 1 in \cite{FominGKLS10tcs}.) The \textsc{Minimum Hitting Set} problem on hypergraphs of fixed 
rank $k$ has also been studied from a parameterized point of view by various authors~\cite{Fernau10dHS,Fernau10weighted,Fernau10,FominGKLS10tcs,NiedermeierR03,Wahlstrom07}. (See Table 2
in \cite{FominGKLS10tcs}.)
A well-known result of Cygan et al. states that for every $c<1$, \textsc{Minimum Hitting Set} cannot be solved by an $O(2^{cn})$ algorithm unless the 
Strong Exponential Time Hypothesis fails~\cite{CyganDLMN12}, while there is a 
$O(2^n)$ algorithm based on verifying all subsets of elements.  
%
The only exponential-time algorithm to enumerate all minimal transversals prior to our work 
had been given by Gaspers in his Master's thesis. 
It is a 
branching algorithm which for hypergraphs of rank $k$ has branching vector $(1,2,\dots,k)$
since in the worst case it selects a hyperedge of size $k$ and recurses on instances with $n-1$, $n-2$, $\dots$, and $n-k$
vertices. The corresponding recurrence for $k=3$ gives a running time of $O(1.8393^n)$.

{\bf Our techniques and results.}
We use various properties of minimal transversals to design branching algorithms and analyze them
using an elaborate Measure \& Conquer analysis.
For an in-depth treatment of branching algorithms, their construction and analysis, branching rules, branching vectors and Measure \& Conquer we refer to \cite{FominK10}.  
For details on our approach, see \autoref{sec:measure} and \cite{Gaspers10,GaspersS12}.
Our main result is the algorithm for hypergraphs of rank $3$ in \autoref{s:rank=3}
which runs in time $O(1.6655^n)$. In \autoref{s:rank=4}  we show that the iterative compression approach from \cite{FominGKLS10tcs} can be extended to enumeration problems and obtain an 
algorithm of running time $O(1.8863^n)$ for hypergraphs of rank $4$. 
In \autoref{s:rank=5} we construct branching algorithms to enumerate the minimal transversals of hypergraphs of rank \longversion{at most }$k$, for all fixed $k\ge 5$. Our algorithmic
results combined with implied upper bounds and new combinatorial lower bounds for the maximum number of minimal transversals in $n$-vertex hypergraphs of rank $k$ are summarized 
in \autoref{tab:mhs}. Finally to underline the potential of our  enumeration algorithms 
let us mention that they can be used to improve upon the best known algorithms   
(see e.g. Table 1 in \cite{FominGKLS10tcs})
for counting minimum hitting sets in hypergraphs of rank $k\ge 3$ and for solving
\textsc{Minimum Hitting Set} in hypergraphs of rank $k\ge 6$.

\begin{table}[tb]  
	\begin{center}
\shortversion{
		\begin{tabular}{c c c c c c}
			\toprule
			Rank \qquad & Lower bound \qquad & Upper bound & Rank \qquad & Lower bound \qquad & Upper bound \\
			\midrule
			2 & $1.4422^n$ & $1.4423^n$ \cite{MoonM65,MillerM60} & 7 & $1.7734^n$ & $O(1.9893^n)$ \\
			3 & $1.5848^n$ & $O(1.6755^n)$ & 8 & $1.7943^n$ & $O(1.9947^n)$ \\
			4 & $1.6618^n$ & $O(1.8863^n)$ & 9 & $1.8112^n$ & $O(1.9974^n)$ \\
			5 & $1.7114^n$ & $O(1.9538^n)$ & 10 & $1.8253^n$ & $O(1.9987^n)$  \\
			6 & $1.7467^n$ & $O(1.9779^n)$ & 20 & $1.8962^n$ & $O(1.9999988^n)$ \\
			\bottomrule
		\end{tabular}
}\longversion{
		\begin{tabular}{c c c}
			\toprule
			Rank \qquad & Lower bound \qquad & Upper bound  \\
			\midrule
			2 & $1.4422^n$ & $1.4423^n$ \cite{MoonM65,MillerM60} \\
			3 & $1.5848^n$ & $O(1.6755^n)$ \\
			4 & $1.6618^n$ & $O(1.8863^n)$ \\
			5 & $1.7114^n$ & $O(1.9538^n)$ \\
			6 & $1.7467^n$ & $O(1.9779^n)$ \\
			7 & $1.7734^n$ & $O(1.9893^n)$ \\
			8 & $1.7943^n$ & $O(1.9947^n)$ \\
			9 & $1.8112^n$ & $O(1.9974^n)$ \\
			10 & $1.8253^n$ & $O(1.9987^n)$  \\
			20 & $1.8962^n$ & $O(1.9999988^n)$ \\
			\bottomrule
		\end{tabular}
}
	\end{center}
	\caption{\label{tab:mhs} Lower and upper bounds for the maximum number of minimal 
                transversals in an $n$-vertex     
                hypergraph of rank $k$. 
               }
\end{table}
{\bf Other related work.}
Enumerating all minimal transversals of a hypergraph is likely to be the most studied problem 
in output-sensitive enumeration. The problem whether there is an output-polynomial
time algorithm, i.e. an algorithm with a running time bounded by a polynomial in the input size and the output size, to enumerate all minimal transversals of a hypergraph is still open
despite efforts of more than thirty years including many of the  leading researchers of 
the field. These efforts have produced many publications on the enumeration of the 
minimal transversals on special hypergraphs~\cite{EiterG03,EiterGM03,ElbassioniR10,%
FredmanK96,KavvadiasS05,KhachiyanBEG07} 
and has also turned output-sensitive enumeration
into an active field of research. 
Recent progress by Kant\'e et al. \cite{KanteLMN14} 
showing that an output-polynomial time algorithm to enumerate all minimal dominating sets of a graph would imply an output-polynomial time algorithm to enumerate the minimal transversals of a hypergraph has triggered a lot of research on the enumeration of minimal dominating sets, both in output-sensitive enumeration and exact exponential 
enumeration; see e.g.~\cite{CouturierHHK13,GolovachHKV15,KanteLMNU15b,KanteLMNU15a}.

\section{Preliminaries}  \label{s:prelim}


We refer to the set of vertices and the set of hyperedges of a hypergraph $H=(V,E)$ 
by $V(H)$ and $E(H)$, respectively.
Throughout the paper we denote the number of vertices of a hypergraph by $n$. 
Note that a hypergraph of rank $k$ has $O(n^k)$ hyperedges.  
Let $v\in V$ be a vertex of $H$.
The \emph{degree} of $v$ in $H$, denoted $d_H(v)$, is the number of hyperedges in $E$ containing $v$.
We omit the subscript when $H$ is clear from the context.
The \emph{neighborhood} of $v$ in $H$, denoted $N(v)$ or $N_H(v)$, is the set of vertices (except $v$ itself) that occur together with $v$ in some hyperedge of $H$.
We denote by $d_i(v)$ or $d_{i,H}(v)$ the number of hyperedges of size $i$ in $H$ containing $v$.
We denote $d_{\le i}(v) := \sum_{j=0}^i d_j(v)$.
If the hypergraph is viewed as a set system $E$ over a ground set $V$, transversals are also called \emph{hitting sets}.
We also say that a vertex $v$ \emph{hits} a hyperedge $e$ if $v\in e$.%
\shortversion{\footnote{Due to space restrictions some proofs 
have been moved to an appendix.}}

%

\section{Hypergraphs of rank 3}  \label{s:rank=3}
\subsection{Measure}
\label{sec:measure}

We will now introduce the measure we use to track the progress of our \longversion{branching }algorithm.
\longversion{\begin{definition}}%
A \emph{measure} $\mu$ for a problem $P$ is a function from the set of all instances for $P$ to the set of
non-negative reals.
\longversion{\end{definition}}%
Our measure will take into account the degrees of the vertices and the number of hyperedges of size 2.
Measures depending on vertex degrees have become relatively standard in the literature.
As for hyperedges of size 2, our measure, similar to Wahlstr{\"o}m's \cite{Wahlstrom04}, indicates an advantage when we can branch on hyperedges of size $2$ once or several times.

Let $H=(V,E)$ be a hypergraph of rank at most $3$.
Denote by $n_k$ the number of vertices of degree $k\in \mathbb{N}$.
Denote by $m_k$ the number of hyperedges of size $k\in \{0,\dots,3\}$.
Also, denote by $m_{\le k} := \sum_{i=0}^k m_i$.
The \emph{measure} of $H$ is
\begin{align*}
 \mu(H) = \Psi(m_{\le 2}) + \sum_{i=0}^{\infty} \omega_i n_i \enspace ,
\end{align*}
where $\Psi: \mathbb{N} \rightarrow \mathbb{R}_{\ge 0}$ is a non-increasing non-negative function such that $\max_i\{\Psi(i)\} = \Psi(0)$ is independent of $n$, and $\omega_i$ are non-negative reals.
The function $\Psi$ is non-increasing since we would like to model that we have an advantage when the size of a hyperedge decreases from 3 to 2.
Clearly, $\mu(H)\ge 0$.

We will now make some assumptions simplifying our analysis and introduce notations easing the description of variations in measure.
We set
\begin{align}
 \omega_i &:= \omega_5\enspace, &
 \Psi(i)  &:= 0        &&\text{for each } i\ge 6, \text{ and}\\
 \Delta \omega_i &:= \omega_i - \omega_{i-1}\enspace, &
 \Delta \Psi(i)  &:= \Psi(i)-\Psi(i-1) &&\text{for each } i\ge 1\enspace.
\intertext{We constrain that}
 0 &\le \Delta \omega_{i+1} \le \Delta \omega_i & \text{and }
 0 &\ge \Delta \Psi(i+1) \ge \Delta \Psi(i) && \text{for each } i\ge 1\enspace. \label{c:deltas}
\end{align}
Note that, by \eqref{c:deltas}, we have that
$\omega_i - \omega_{i-k} \ge k \cdot \Delta \omega_i$ for $0 \le k\le i$.
In addition, our branching rules will add constraints on the measure.
Note that a branching rule with one branch is a \emph{reduction} rule and one with no branch is a \emph{halting} rule. 
Denote by $T(\mu) := 2^{\mu}$ an upper bound on the number of leaves of the search trees modeling the recursive calls of the algorithm for all 
$H$ with $\mu(H) \le \mu$.

Suppose a branching rule creates $k\ge 1$ branches $B[1], \dots, B[k]$, each $B[i]$ decreasing the measure by $\eta_i$.
Then, we obtain the following constraint on the measure:
\longversion{
	\begin{align}
     \sum_{i=1}^k T(\mu - \eta_i) &\le T(\mu) \enspace.\nonumber\\
     \intertext{Dividing by $2^\mu$, the constraint becomes}
     \sum_{i=1}^k 2^{- \eta_i} &\le 1 \enspace. \label{c:main}
    \end{align}
}\shortversion{
 $\sum_{i=1}^k T(\mu - \eta_i) \le T(\mu).$
 Dividing by $2^\mu$, the constraint becomes
 \begin{align}
  \sum_{i=1}^k 2^{- \eta_i} &\le 1 \enspace. \label{c:main}
 \end{align}
}
Given these constraints for all branching rules, we will determine values for $\Psi(0), \dots, \Psi(5), \omega_0, \dots, \omega_5$ so as to minimize the maximum value of $\mu(H)/|V(H)|$ taken over all rank-3 hypergraphs $H$ when $|V(H)|$ is large.
Since $\Psi(m_{\le 2})$ is a constant, this part of the measure contributes only a constant factor to the running time.
Given our assumptions on the weights, optimizing the measure amounts to solving a convex program \cite{Gaspers10,GaspersS12} minimizing $\omega_5$ subject to all constraints, which can be done with certificates
of optimality.
If\longversion{, in addition,} we make sure that the maximum recursion depth of the algorithm is polynomial\longversion{ly bounded}, we obtain that the running time is within a polynomial factor of
$2^{\omega_5}$.

\subsection{Algorithm} 
\label{sec:algo3}

An instance of a recursive call of the algorithm is a hypergraph $H=(V,E)$ with rank at most 3 and a set $S$, which is a partial hitting set for the original hypergraph.
The hypergraph $H$ contains all hyperedges that still need to be hit, and the vertices that are eligible to be added to $S$.
Thus, $V\cap S = \emptyset$.
Initially, $S = \emptyset$.
Each branching rule has a condition which is a prerequisite for applying the rule.
When the prerequisites of more than one rule hold, the first applicable rule is used.
Our branching rules create subinstances where some vertices are selected and others are discarded.
\begin{itemize}
 \item If we \emph{select} a vertex $v$, we remove all hyperedges containing $v$ from the subinstance of the branch, we add $v$ to $S$, and remove $v$ from $V$.
 \item If we \emph{discard} a vertex $v$, we remove $v$ from all hyperedges and from $V$.
\end{itemize}

\noindent
We now come to the description of the branching rules, their correctness, and their analysis, i.e., the constraints they impose on the measure.
Rules $0.x$ are halting rules, and
Rules $1.x$ are reduction rules.
Each rule first states the prerequisite, then describes
its actions, then the soundness is proved if necessary, and the constraints on the measure are given for the analysis.

\longversion{
 \newcommand{\myrule}[1]{\par\medskip \noindent \textbf{Rule #1} }
}\shortversion{
 \newcommand{\myrule}[1]{\par\smallskip \noindent \textbf{Rule #1} }
}

\myrule{0.0}
There is a hyperedge $e\in E$ of size 0.
Do nothing.
The algorithm backtracks and enumerates no transversal in this recursive call
since there remains a hyperedge that cannot be hit by a vertex.

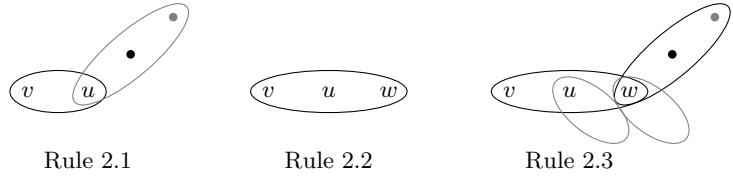
\myrule{0.1}
$E=\emptyset$.
If $S$ is a minimal transversal for the original hypergraph, output $S$, otherwise do nothing.
We are in a recursive call where no hyperedge remains to be hit.
Thus, $S$ is a transversal of the original input hypergraph.
However, $S$ might not be minimal, and the algorithm checks in polynomial time 
if $S$ is a minimal transversal of the initial hypergraph and outputs it if so.

\myrule{1.0}
There is a vertex $v \in V$ with degree 0.
Discard $v$.
Indeed, since $v$ hits no hyperedge of $H$, a transversal for $H$ containing $v$ is not minimal.
The constraint on the measure is $2^{-\omega_0} \le 2^0$, which is trivially satisfied since $\omega_0\ge 0$.

\myrule{1.1}
There is a hyperedge $e_1 \in E$ that is a subset of another hyperedge $e_2 \in E$ of size $3$.
Remove $e_2$ from $E$.
The rule is sound since each transversal of $(V,E\setminus \{e_2\})$ is also a transversal of $H$.
This is because any vertex that hits $e_1$ also hits $e_2$.
Since $e_2$ has size $3$, this rule has no effect on the measure; the constraint $2^0 \le 2^0$ is always satisfied.

\myrule{1.2}
There is a hyperedge $e$ of size 1.
Select $v$, where $\{v\} = e$.
The rule is sound since $v$ is the only vertex that hits $e$.
Selecting $v$ removes $v$ and all hyperedges containing $v$ from the instance.
By \eqref{c:deltas}, the decrease in measure is at least $\eta_1 = \omega_{d(v)} + \Psi(d_{\le 2}(v)) - \Psi(0)$.
To fulfill \eqref{c:main}, we will need to constrain that $\eta_1 \ge 0$.
Since $d_{\le 2}(v) \le d(v)$ and by \eqref{c:deltas}, if suffices to constrain
\begin{align}
 \Psi(i)-\Psi(0) &\ge - \omega_i &&\text{for } 1\le i\le 6\enspace.
\end{align}

\myrule{2}
There is a vertex $v\in V$ with degree one.
Denote $e$ the unique hyperedge containing $v$ and branch according to the following three subrules.

\tikzstyle{vertex}=[minimum size=1mm,circle,fill=black,inner sep=0mm,draw=black]
\tikzstyle{optvertex}=[minimum size=1mm,circle,fill=black!50,inner sep=0mm,draw=black!50]
  
\begin{figure}[tbp]
 \centering
 \begin{tikzpicture}[xscale=0.8,yscale=0.7]
  \draw (0,0) node {$v$};
  \draw (1,0) node {$u$};
  \draw ([shift=(45:1cm)] 1,0) node[vertex] {};
  \draw ([shift=(45:2cm)] 1,0) node[optvertex] {};
  \draw (0.5,0) ellipse (0.8cm and 0.4cm);
  \draw[draw=black!50,rotate around={45:([shift=(45:1cm)] 1,0)}] ([shift=(45:1cm)] 1,0) ellipse (1.3cm and 0.4cm);
  \draw (1,-1.3) node {Rule 2.1};
  
 \begin{scope}[xshift=4cm]
  \draw (0,0) node {$v$};
  \draw (1,0) node {$u$};
  \draw (2,0) node {$w$};
  \draw (1,0) ellipse (1.3cm and 0.4cm);
  \draw (1,-1.3) node {Rule 2.2};
 \end{scope}
 
 \begin{scope}[xshift=8cm]
  \draw (0,0) node {$v$};
  \draw (1,0) node {$u$};
  \draw (2,0) node {$w$};
  \draw ([shift=(45:1cm)] 2,0) node[vertex] {};
  \draw ([shift=(45:2cm)] 2,0) node[optvertex] {};
  \draw (1,0) ellipse (1.3cm and 0.4cm);
  \draw[rotate around={45:([shift=(45:1cm)] 2,0)}] ([shift=(45:1cm)] 2,0) ellipse (1.3cm and 0.4cm);
  \draw[draw=black!50,rotate around={-45:([shift=(-45:0.5cm)] 1,0)}] ([shift=(-45:0.5cm)] 1,0) ellipse (0.8cm and 0.4cm);
  \draw[draw=black!50,rotate around={-45:([shift=(-45:0.5cm)] 2,0)}] ([shift=(-45:0.5cm)] 2,0) ellipse (0.8cm and 0.4cm);
  \draw (1,-1.3) node {Rule 2.3};
 \end{scope}
 \end{tikzpicture}
\end{figure}

\myrule{2.1}
The hyperedge $e$ has size $2$. Denote $e=\{v,u\}$.
Branch into two subproblems: $B[1]$ where $v$ is selected, and $B[2]$ where $v$ is discarded.
In $B[1]$, the vertex $u$ is discarded due to minimality of the transversal, and the number of hyperedges of size at most $2$ decreases by $1$ since
$e$ is removed. By \eqref{c:deltas} the decrease in measure in $B[1]$ is at least $\eta_1 = \omega_{1} + \omega_{d(u)} + \Delta\Psi(m_{\le 2})$.
In $B[2]$, the vertex $u$ is selected by applying Rule 1.2 after having discarded $v$.
We have that $d_{\le 2}(u) \le \min(d(u),m_{\le 2})$ hyperedges of size at most 2 disappear.
By \eqref{c:deltas} the decrease in measure in $B[2]$ is at least
$\eta_2 = \omega_{1} + \omega_{d(u)} + \Psi(m_{\le 2}) - \Psi(\max(m_{\le 2}-d(u),0))$.
Note that we do not take into account additional sets of size at most $2$ that may be created by discarding $u$ and the degree-decreases of the other vertices in the neighborhood of $u$ as a result of selecting $u$.
However, these do not increase the measure due to constraints \eqref{c:deltas}.
Thus, we constrain that
\begin{align}
 2^{-\omega_1 - \omega_{d(u)} - \Delta\Psi(m_{\le 2})}
 +2^{-\omega_1 - \omega_{d(u)} - \Psi(m_{\le 2}) + \Psi(\max(m_{\le 2}-d(u),0))}
 &\le 1\enspace, \label{c:21}
\end{align}
for \longversion{all $d(u)$ and $m_{\le 2}$ such that }$1\le d(u) \le 6$ and $1\le m_{\le 2}\le 6$. Note that the value of $d(u)$ ranges up to 6 instead of 5 although $\omega_5=\omega_6$, so that we have a constraint modeling that the value of $\Psi$ increases from $\Psi(6)=0$ to $\Psi(0)$ in the second branch.

In the remaining subrules of Rule 2, the hyperedge $e$ has size $3$.

\myrule{2.2}
The other two vertices in $e$ also have degree 1.
Branch into 3 subproblems adding exactly one vertex from $e$ to $S$ and discarding the other 2.
Clearly, any minimal transversal contains exactly one vertex from $e$.
The decrease in measure is $3\omega_1$ in each branch, giving the constraint
\begin{align}
 3 \cdot 2^{-3\omega_1}
 & \le 1\enspace.
\end{align}

\myrule{2.3}
Otherwise, let $e = \{v,u,w\}$ with $d(u)\ge 2$.
We create two subproblems: in $B[1]$ we select $v$ and in $B[2]$ we discard $v$.
Additionally, in $B[1]$ we discard $u$ and $w$ due to minimality.
The measure decreases by at least
$\eta_1 = 2\omega_1+\omega_2$ and
$\eta_2 = \omega_1 - \max_{i\ge 1} \{\Delta\Psi(i)\} = \omega_1$; the last equality holds since $\Psi(i)=0, i\ge 6$, and by \eqref{c:deltas}.
We obtain the constraint
\begin{align}
 2^{-2\omega_1-\omega_2}
 +2^{-\omega_1}
 \le 1\enspace.
\end{align}

\myrule{3}
At least one hyperedge has size 2.
Let $v$ be a vertex that has maximum degree among all vertices contained in a maximum number of hyperedges of size 2.
Let $e=\{v,u_1\}$ be a hyperedge containing $v$.
Note that, due to Rule 2, every vertex has degree at least 2\longversion{ if Rule 3 applies}.
We branch according to the following subrules.

\begin{figure}[tbp]
 \centering
 \begin{tikzpicture}[xscale=0.8,yscale=0.7]
  \draw (0,0) node {$v$};
  \draw (1,0) node {$u_1$};
  \draw ([shift=(45:1cm)] 0,0) node[vertex] {};
  \draw ([shift=(45:2cm)] 0,0) node[vertex] {};
  \draw ([shift=(45:1cm)] 1,0) node[vertex] {};
  \draw ([shift=(45:2cm)] 1,0) node[vertex] {};
  \draw (0.5,0) ellipse (0.8cm and 0.35cm);
  \draw[rotate around={45:([shift=(45:1cm)] 0,0)}] ([shift=(45:1cm)] 0,0) ellipse (1.3cm and 0.35cm);
  \draw[rotate around={45:([shift=(45:1cm)] 1,0)}] ([shift=(45:1cm)] 1,0) ellipse (1.3cm and 0.35cm);
  \draw (1,-1.3) node {Rule 3.1};
  
 \begin{scope}[xshift=4cm]
  \draw (0,0) node {$v$};
  \draw ([shift=(-45:1cm)] 0,0) node {$u_2$};
  \draw (1,0) node {$u_1$};
  \draw ([shift=(45:1cm)] 0,0) node[vertex] {};
  \draw ([shift=(45:2cm)] 0,0) node[vertex] {};
  \draw[draw=black!50,rotate around={45:([shift=(45:1cm)] 0,0)}] ([shift=(45:1cm)] 0,0) ellipse (1.3cm and 0.35cm);
  \draw (0.5,0) ellipse (0.8cm and 0.35cm);
  \draw[rotate around={-45:([shift=(-45:0.5cm)] 0,0)}] ([shift=(-45:0.5cm)] 0,0) ellipse (0.8cm and 0.35cm);
  \draw (1,-1.3) node {Rule 3.2};
 \end{scope}
 
 \begin{scope}[xshift=8cm]
  \draw (0,0) node {$v$};
  \draw ([shift=(-45:1cm)] 0,0) node {$u_3$};
  \draw (1,0) node {$u_2$};
  \draw ([shift=(45:1cm)] 0,0) node {$u_1$};
  \draw[draw=black!50,rotate around={90:([shift=(90:0.5cm)] 0,0)}] ([shift=(90:0.5cm)] 0,0) ellipse (0.8cm and 0.35cm);
  \draw[rotate around={45:([shift=(45:0.5cm)] 0,0)}] ([shift=(45:0.5cm)] 0,0) ellipse (0.8cm and 0.35cm);
  \draw (0.5,0) ellipse (0.8cm and 0.35cm);
  \draw[rotate around={-45:([shift=(-45:0.5cm)] 0,0)}] ([shift=(-45:0.5cm)] 0,0) ellipse (0.8cm and 0.35cm);
  \draw (1,-1.3) node {Rule 3.3};
 \end{scope}
 \end{tikzpicture}
\end{figure}

\myrule{3.1}
$d_2(v)=1$.
We branch into two subproblems: in $B[1]$ we select $v$ and in $B[2]$ we discard $v$.
Additionally, in $B[2]$, we select $u_1$ by Rule 1.2.
Since $u_1$ is contained in the hyperedge $e$ and since $d_2(u_1) \le d_2(v)$, we have that $d_2(u_1)=1$ and therefore $d(v) \ge d(u_1)$.
In $B[1]$, we observe that the degrees of $v$'s neighbors decrease. Also, $e$ is a hyperedge of size $2$ and it is removed; this affects the value of $\Psi$. The measure decrease in the first branch is therefore at least
$\omega_{d(v)} + \Delta\omega_{d(u_1)} + \Delta\Psi(m_{\le 2}) \ge \omega_{d(u_1)} + \Delta\omega_{d(u_1)} + \Delta\Psi(m_{\le 2})$.
In $B[2]$, since we select $u_1$, we have that $d_2(u_1)=1$ hyperedge of size 2 disappears, and since we discard $v$, we have that $d(v)-1$ sets of size 3 will become sets of size 2.
Also, none of these size-2 sets already exist in $E$, otherwise Rule 1.1 would apply.
We have a measure decrease of $\omega_{d(v)} + \omega_{d(u_1)} +\Psi(m_{\le 2}) -\Psi(m_{\le 2}+d(v)-2) \ge 2\omega_{d(u_1)} +\Psi(m_{\le 2}) -\Psi(m_{\le 2}+d(u_1)-2)$.
We obtain the following set of constraints:
\begin{align}
 2^{-\omega_{d(u_1)} - \Delta\omega_{d(u_1)} - \Delta\Psi(m_{\le 2})}
 +2^{-2\omega_{d(u_1)} -\Psi(m_{\le 2}) +\Psi(m_{\le 2}+d(u_1)-2)}
 \le 1\enspace, \label{c:31}
\end{align}
for each $d(u_1)$ and $m_{\le 2}$ with $2\le d(u_1) \le 6$ and $1\le m_{\le 2}\le 6$.
Here, $d(u_1)$ ranges up to 6 since $\Delta\omega_6=0$, whereas $\Delta\omega_5$ may be larger than 0.

\myrule{3.2}
$d_2(v)=2$.
We branch into two subproblems: in $B[1]$ we select $v$ and in $B[2]$ we discard $v$.
Denoting $\{v,u_2\}$ the second hyperedge of size $2$ containing $v$,
we additionally select $u_1$ and $u_2$ in $B[2]$.
In $B[1]$, the measure decrease is at least $\eta_1 = \omega_{d(v)} + \Delta\omega_{d(u_1)} + \Delta\omega_{d(u_2)} + \Psi(m_{\le 2}) - \Psi(m_{\le 2}-2)$.
In $B[2]$, selecting $u_1$ and $u_2$ removes $d_2(u_1)+d_2(u_2)\le \min(4,m_{\le 2})$ hyperedges of size $2$, and discarding $v$ decreases the size of $d(v)-2$ hyperedges from 3 to 2.
Thus, the measure decrease is at least $\eta_2 = \omega_{d(v)} + \omega_{d(u_1)} + \omega_{d(u_2)} + \Psi(m_{\le 2}) - \Psi(\max(m_{\le 2}-4,0) +d(v)-2)$.
We obtain the following set of constraints:
\begin{align}
 &2^{-\omega_{d(v)} - \Delta\omega_{d(u_1)} - \Delta\omega_{d(u_2)} -\Psi(m_{\le 2}) + \Psi(m_{\le 2}-2)} \nonumber\\
 +\;&2^{-\omega_{d(v)} - \omega_{d(u_1)} - \omega_{d(u_2)} - \Psi(m_{\le 2}) + \Psi(\max(m_{\le 2}-4,0) +d(v)-2)}  \longversion{\nonumber\\ }
 \le\; \longversion{&} 1\enspace, \label{c:32}
\end{align}
for \longversion{each $d(v), d(u_1), d(u_2)$, and $m_{\le 2}$ with }$2\le d(v), d(u_1), d(u_2) \le 6$ and $2\le m_{\le 2}\le 6$.

\myrule{3.3}
$d_2(v) \geq 3$.
We branch into two subproblems: in $B[1]$ we select $v$ and in $B[2]$ we discard $v$.
In $B[2]$ we additionally select all vertices occurring in hyperedges of size 2 with $v$.
Denote by $\{v,u_2\}$ and $\{v,u_3\}$ a second and third hyperedge of size $2$ containing $v$.
In $B[1]$, the number of size-2 hyperedges decreases by $d_2(v)$. The measure decrease is at least $\omega_{d_2(v)} + \Delta\omega_{d(u_1)} + \Delta\omega_{d(u_2)} + \Delta\omega_{d(u_3)} + \Psi(m_{\le 2}) - \Psi(m_{\le 2}-d_2(v))$.
In $B[2]$, the number of hyperedges of size at most 2 decreases at most by $d_2(u_1)+d_2(u_2)+d_2(u_3) \le \min(d(u_1)+d(u_2)+d(u_3),m_{\le 2})$.
We obtain a measure decrease of at least $\omega_{d_2(v)} + \omega_{d(u_1)} + \omega_{d(u_2)} + \omega_{d(u_3)} + (d_2(v)-3)\cdot \omega_2 + \Psi(m_{\le 2}) - \Psi(\max(m_{\le 2}-d(u_1)-d(u_2)-d(u_3),0))$.
The family of constraints for this branching rule is therefore
\begin{align}
 &2^{-\omega_{d_2(v)} -\sum_{i=1}^3 \Delta\omega_{d(u_i)} - \Psi(m_{\le 2}) + \Psi(m_{\le 2}-d_2(v))}\nonumber\\
 +\;&2^{-\omega_{d_2(v)} -\sum_{i=1}^3 \omega_{d(u_i)} - (d_2(v)-3)\cdot\omega_2 -\Psi(m_{\le 2}) + \Psi(\max(m_{\le 2}-\sum_{i=1}^3 d(u_i),0))} \longversion{\nonumber\\}
 \le\; \longversion{&} 1\enspace,
\end{align}
where $3\le d_2(v) \le m_{\le 2} \le 6$ and $2\le d(u_1),d(u_2),d(u_3)\le 6$.

\myrule{4}
Otherwise, all hyperedges have size 3.
Choose $v\in V$ with maximum degree, and branch according to the following subrules.

\begin{figure}[tbp]
 \centering
 \begin{tikzpicture}[xscale=0.8,yscale=0.7]
  \draw (0,0) node {$v$};
  \draw (1,0) node[vertex] {};
  \draw (2,0) node[vertex] {};
  \draw ([shift=(90:1cm)] 0,0) node[vertex] {};
  \draw ([shift=(90:2cm)] 0,0) node[vertex] {};
  \draw ([shift=(45:1cm)] 0,0) node[vertex] {};
  \draw ([shift=(45:2cm)] 0,0) node[vertex] {};
  \draw ([shift=(-45:1cm)] 0,0) node[vertex] {};
  \draw ([shift=(-45:2cm)] 0,0) node[vertex] {};
  \draw[draw=black!50,rotate around={90:([shift=(90:1cm)] 0,0)}] ([shift=(90:1cm)] 0,0) ellipse (1.3cm and 0.35cm);
  \draw[rotate around={45:([shift=(45:1cm)] 0,0)}] ([shift=(45:1cm)] 0,0) ellipse (1.3cm and 0.35cm);
  \draw (1,0) ellipse (1.3cm and 0.35cm);
  \draw[rotate around={-45:([shift=(-45:1cm)] 0,0)}] ([shift=(-45:1cm)] 0,0) ellipse (1.3cm and 0.35cm);
  \draw (1,-1.9) node {Rule 4.1};
  
 \begin{scope}[xshift=4cm]
  \draw (0,0) node {$v$};
  \draw (1,0) node {$u$};
  \draw ([shift=(45:1cm)] 1,0) node[vertex] {};
  \draw ([shift=(-45:1cm)] 1,0) node[vertex] {};
  \draw[rotate around={22.5:([shift=(22.5:1cm)] 0,0)}] ([shift=(22.5:1cm)] 0,0) ellipse (1.3cm and 0.6cm);
  \draw[rotate around={-22.5:([shift=(-22.5:1cm)] 0,0)}] ([shift=(-22.5:1cm)] 0,0) ellipse (1.3cm and 0.6cm);
  \draw (1,-1.9) node {Rule 4.2};
 \end{scope}
 
 \begin{scope}[xshift=8cm]
  \draw (0,0) node {$v$};
  \draw ([shift=(30:1cm)] 0,0) node {$u_1$};
  \draw ([shift=(30:2cm)] 0,0) node {$w_1$};
  \draw ([shift=(-30:1cm)] 0,0) node {$u_2$};
  \draw ([shift=(-30:2cm)] 0,0) node {$w_2$};
  \draw[rotate around={30:([shift=(30:1cm)] 0,0)}] ([shift=(30:1cm)] 0,0) ellipse (1.3cm and 0.45cm);
  \draw[rotate around={-30:([shift=(-30:1cm)] 0,0)}] ([shift=(-30:1cm)] 0,0) ellipse (1.3cm and 0.45cm);
  \draw (1,-1.9) node {Rule 4.3};
 \end{scope}
 \end{tikzpicture}
\end{figure}

\myrule{4.1}
$d(v) \geq 3$.
We branch into two subproblems: in $B[1]$ we select $v$ and in $B[2]$ we discard $v$.
In $B[1]$, the degree of each of $v$'s neighbors decreases by the number of hyperedges it shares with $v$.
We obtain a measure decrease of at least
$\omega_{d(v)} + \sum_{u\in N(v)} (\omega_{d(u)} - \omega_{d(u)-|\{e\in E: \{u,v\} \subseteq e|})$, which, by  \eqref{c:deltas}, is at least $\omega_{d(v)} + 2 \cdot d(v) \cdot \Delta\omega_{d(v)}$.
In $B[2]$, the number of hyperedges of size 2 increases from 0 to $d(v)$, for a measure decrease of at least $\omega_{d(v)} + \Psi(0) - \Psi(d(v))$.
For $d(v)\in \{3,\dots,6\}$, this branching rules gives the constraint
\begin{align}
2^{-\omega_{d(v)} - 2 \cdot d(v) \cdot \Delta\omega_{d(v)}}
+2^{-\omega_{d(v)} - \Psi(0) + \Psi(d(v))}
\le 1\enspace. \label{c:41}
\end{align}

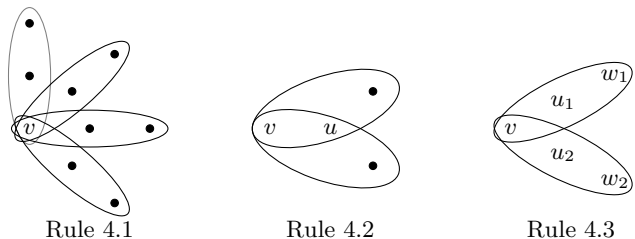
\myrule{4.2}
$d(v)=2$ and there is a vertex $u\in V\setminus \{v\}$ which shares two hyperedges with $v$.
We branch into two subproblems: in $B[1]$ we select $v$ and in $B[2]$ we discard $v$.
Additionally, we discard $u$ in $B[1]$ due to minimality.
Since each vertex has degree 2, we obtain the following constraint.
\begin{align}
 2^{-2\omega_{2} - 2 \Delta\omega_{2}}
 +2^{-\omega_{2} - \Psi(0) + \Psi(2)}
 \le 1\enspace.
\end{align}

\myrule{4.3}
Otherwise, $d(v)=2$ and for every two distinct $e_1,e_2 \in E$ with $v\in e_1$ and $v\in e_2$ we have that $e_1 \cap e_2 = \{v\}$.
Denoting $e_1=\{v,u_1,w_1\}$ and $e_2=\{v,u_2,w_2\}$ the two hyperedges containing $v$, we branch into three subproblems: in $B[1]$ we select $v$ and $u_1$; in $B[2]$ we select $v$ and discard $u_1$; and in $B[3]$ we discard $v$.
Additionally, we discard $u_2$ and $w_2$ in $B[1]$ due to minimality.
Again, all vertices have degree $2$ in this case.
In $B[1]$, the degree of $w_1$ decreases to 1. Among the two hyperedges containing $u_2$ and $w_2$ besides $e_2$, at most one is hit by $u_1$, since $d(u_1)=2$, and none of them is hit by $v$; thus, the branch creates at least one hyperedge of size at most 2. The measure decrease is at least $4\omega_2+\Delta\omega_2-\Delta\Psi(1)$.
In $B[2]$, the degrees of $w_1, u_2$, and $w_2$ decrease by 1 and the size of a hyperedge containing $u_1$ decreases.
The measure decrease is at least $2\omega_2+3\Delta\omega_2-\Delta\Psi(1)$.
In $B[3]$, two size-2 hyperedges are created for a measure decrease of $\omega_2-\Delta\Psi(2)$.
This gives us the following constraint:
\begin{align}
 2^{-4\omega_{2} - \Delta\omega_{2} + \Delta\Psi(1)}
 +2^{-2\omega_{2} - 3\Delta\omega_2 + \Delta\Psi(1)}
 +2^{-\omega_{2} + \Delta\Psi(2)}
 \le 1\enspace.
\end{align}

\noindent
This finishes the description of the algorithm.
We \shortversion{can now}\longversion{are now ready to} prove an upper bound on its running time, along the lines described in \autoref{sec:measure}\longversion{, using the following lemma}.

\begin{lemma}[\cite{GaspersS12}, Lemma 2.5 in \cite{Gaspers10}] \label{lem:measureanalysis}
Let $A$ be an algorithm for a problem $P$,
$c \ge 0$ be a constant, and $\mu(\cdot), \eta(\cdot)$ be measures
for the instances of $P$,
such that for any input instance $I$,
Algorithm $A$ transforms $I$ into instances $I_1,\ldots,I_k$,
solves these recursively, and combines their solutions to solve~$I$,
using time $O(\eta(I)^{c})$ for the transformation and combination steps
(but not the recursive solves), and
\begin{align}
(\forall i) \quad \eta(I_i) & \leq \eta(I)-1\enspace\text{, and}  \label{eq:masize}
  \\
\sum_{i=1}^k 2^{\mu(I_i)} & \leq 2^{\mu(I)}\enspace. \label{eq:magain}
\end{align}
Then $A$ solves any instance $I$
in time $O(\eta(I)^{c+1}) 2^{\mu(I)}$.
\end{lemma}

\begin{theorem}\label{thm:r3}
The described algorithm enumerates all minimal transversals of an $n$-vertex hypergraph of
rank $3$ in time $O(1.6755^n)$.
\end{theorem}
\begin{proof}
Consider any input instance $H=(V,E)$ with $n$ vertices and measure $\mu = \mu(H)$.
Using the following weights, the measure $\mu$ satisfies all constraints.
\begin{center}
\begin{tabular}[t]{l l l}
$i$ \hspace{0.5cm} & $\omega_i$  \hspace{2cm} & $\Psi(i)$\\
\hline
$0$ & $0$ & $0.566096928$\\
$1$ & $0.580392137$ & $0.436314617$\\
$2$ & $0.699175718$ & $0.306532603$\\
$3$ & $0.730706814$ & $0.211986294$\\
\end{tabular}\hspace{0.5cm}
\begin{tabular}[t]{l l l}
$i$ \hspace{0.5cm} & $\omega_i$  \hspace{2cm} & $\Psi(i)$\\
\hline
$4$ & $0.742114220$ & $0.119795899$\\
$5$ & $0.744541491$ & $0.035202514$\\
$6$ & $0.744541491$ & $0$\\
\end{tabular}
\end{center}
Also, $\mu \le \omega_5 \cdot n + \Psi(0)$.
Therefore, the number of times a halting rule is executed (i.e., the number of leaves of the search tree) is at most $2^{\omega_5 \cdot n + O(1)}$.
Since each recursive call of the algorithm decreases the measure $\eta(H):=|V|+|E|$ by at least $1$, the height of the search tree is polynomial.
We conclude, by \autoref{lem:measureanalysis}, that the algorithm has running time $O(1.6755^n)$ since $2^{\omega_5} = 1.6754..$.
\qed
\end{proof}

\noindent
In the proof of \autoref{thm:r3}, the tight constraints are
\eqref{c:21} with $(d(u),m_{\le 2})\in \{(5,5),(6,5),(6,6)\}$,
\eqref{c:31} with $d(u_1)=2$ and $m_{\le 2}=1$,
\eqref{c:32} with $d(v)=2$, $d(u_1)=d(u_2)=6$, and $2\le m_{\le 2}\le 4$, and
\eqref{c:41} for all values of $d(v)$.


\section{Hypergraphs of rank 4}  \label{s:rank=4}

For hypergraphs of rank $4$, we adapt an iterative compression algorithm of \cite{FominGKLS10tcs}, which was designed for counting the number of 
minimum transversals, to the enumeration setting.

\begin{theorem}\label{thm:ic}  
	Suppose there is an algorithm with running time $O^*((a_{k-1})^n)$,
	$1 < a_{k-1} \le 2$,
	enumerating all minimal transversals in rank-$(k-1)$ hypergraphs.
	Then all minimal transversals in a rank-$k$ hypergraph can be enumerated in time
	\[
	\min_{0.5 \le \alpha \le 1} \max \left\{ O^*\left({n \choose \alpha n}\right),
	O^*\left(2^{\alpha n} (a_{k-1})^{n-\alpha n} \right) \right\}.
	\]
\end{theorem}
\newcommand{\proofthmic}{%
\begin{proof}
	Let $H=(V,E)$ be a rank-$k$ hypergraph.
	First the algorithm tries  all subsets of $V$ of size at least
	$\left\lfloor  \alpha n \right\rfloor$
	and outputs those that are minimal transversals of $H$.
	The running time of this phase is $O^*\left(\sum_{i=\left\lfloor \alpha n \right\rfloor}^{n} {n \choose i}\right)
	=O^*\left({n \choose \alpha n}\right)$.
	
	Now there are two cases. In the first case, there is no transversal of size $\left\lfloor  \alpha n \right\rfloor$.
	But then, $H$ has no transversals of size at most $\left\lfloor \alpha n \right\rfloor$ and we are done.
	
	In the second case, there exists a transversal $X$ of size $\lfloor\alpha n \rfloor$.
	For each subset $N\subseteq X$, the algorithm will enumerate all minimal transversals of $H$ whose intersection with $X$ is $N$.
	For a given subset $N\subseteq X$, obtain a new hypergraph $H'=(V',E')$ from $H$ by removing the hyperedges that contain a vertex from $N$, removing the vertices $X\setminus N$ from all remaining hyperedges, and setting $V'=V\setminus X$.
	Observe that for every minimal transversal $X'$ of $H$ with $X'\cap X=N$, we have that $X'\setminus N$ is a minimal transversal of $H'$.
	Moreover, observe that the rank of $H'$ is at most $k-1$, since all vertices from $X$ have been removed and they form a transversal of $H$.
	Therefore, the algorithm invokes an algorithm for enumerating all minimal transversals in rank-$(k-1)$ hypergraphs and runs it on $H'$.
	For each minimal transversal $Y$ that is output, check whether $X\cup Y$ is a minimal
	transversal of $H$, and if so, output $X\cup Y$.
	This phase of the algorithm has running time 
	$O^*\left(2^{\alpha n} (a_{k-1})^{n-\alpha n} \right)$.
	\qed
\end{proof}	
}
\longversion{\proofthmic}

\noindent
\shortversion{Combined with \autoref{thm:r3}}\longversion{Combining \autoref{thm:ic} with \autoref{thm:r3}}, the running time is minimized for $\alpha \approx 0.66938$.

\begin{theorem}\label{thm:r4}
	The described algorithm enumerates all minimal transversals of an $n$-vertex 
hypergraph of rank $4$ in time $O(1.8863^n)$. 
\end{theorem}

\section{Hypergraphs of rank at least 5}  \label{s:rank=5}

For a hypergraph $H=(V,E)$ of rank $k\ge 5$, we use the following algorithm to enumerate all minimal transversals.
As in \autoref{sec:algo3}, the instance of a recursive call of the algorithm is a hypergraph $H=(V,E)$ and set $S$ which is a partial transversal of the original hypergraph. 
The hypergraph $H$ contains all hyperedges that still need to be hit and the vertices that can be added to $S$. 
The algorithm  enumerates all minimal transversals $Y$ of the original hypergraph such that $Y\setminus S$ is a minimal transversal of $H$.
\begin{itemize}
\item[H1] If $E=\emptyset$, then check whether $S$ is a minimal transversal of the original hypergraph and output $S$ if so.
\item[H2] If $\emptyset\in E$, then $H$ contains an empty hyperedge, and the algorithm backtracks.
\item[R1] If there is a vertex $v\in V$ with $d_{H}(v)=0$, then discard $v$ and recurse.
\item[R2] If there are two hyperedges $e_1,e_2\in E$ with $e_1\subseteq e_2$, then remove $e_2$ and recurse.
\item[R3] If there is a hyperedge $e\in E$ with $|e|=1$, then select $v$, where $e=\{v\}$ and recurse.
\item[B1] If there is a vertex $v\in V$ with $d_{H}(v)=1$, then let $e\in E$ denote the hyperedge with $v\in e$. Make one recursive call where $v$ is discarded, and one recursive call where $v$ is selected and all vertices from $e\setminus \{v\}$ are discarded.
\item[B2] Otherwise, select two hyperedges $e,e'$ such that $e$ is a smallest hyperedge and $|e\cap e'|\ge 1$. Order their vertices such that their common vertices appear first: $e=\{v_1,\dots,v_{|e|}\}$ and $e'=\{v_1,\dots,v_\ell,u_{\ell+1},\dots,u_{|e'|}\}$.
Make $|e|$ recursive calls; in the $i$th recursive call, $v_1,\dots,v_{i-1}$ are discarded and $v_i$ is selected.
\end{itemize}

\begin{theorem}\label{thm:r5}
	For any $k\ge 2$, \longversion{the described algorithm enumerates }the minimal transversals of a
    rank-$k$ hypergraph \shortversion{can be enumerated }in time $O((\beta_k)^n)$, where $\beta_k$ is the positive real root of
	\begin{align*}
	-1+x^{-1}+\sum_{i=3}^{k} (i-2)\cdot x^{-i} + \sum_{i=k+1}^{2k-1}(2k-i)\cdot x^{-i} = 0\enspace.
	\end{align*}
\end{theorem}
\begin{proof}
	The correctness of the halting and reduction rules are easy to see.
	For the correctness of branching rule B1, it suffices to observe that a transversal containing $v$ and some other vertex from $e$ is not minimal.
	The correctness of B2 follows since the $i$th recursive call enumerates the minimal transversals such that the first vertex among $v_1,\dots,v_{|e|}$ they contain is $v_i$.
	
	As for the running time, a crude analysis gives the same running time as the
	analysis of \cite{Gaspers05}, since
	we can associate the branching vector $(1,|e|)$ with B1, which is $(1,2)$ in the worst case,
	and the branching vector $(1,2,\dots,k)$ with B2.
	
	Let us look at B2 more closely. In the worst case, $|e|=k$.
	Due to the reduction rules, we have that $|e\cap e'|=\ell<|e|$.
\begin{figure}[tb]
	\begin{center}
		\begin{tikzpicture}[scale=0.8,sibling distance=0pt, level distance=35pt]
		\tikzstyle{every node}=[minimum size=8mm]
		\tikzset{edge from parent/.style={draw, edge from parent path=
				{(\tikzparentnode) -- (\tikzchildnode)}}}
		\Tree [.0 1 2 $\dots$ $\ell$ [.$\ell+1$ {\footnotesize $\ell+2$} {\footnotesize $\ell+3$} {\footnotesize $\dots$} {\footnotesize $k+1$} ] [.$\ell+2$ {\footnotesize $\ell+3$} {\footnotesize $\ell+4$} {\footnotesize $\dots$} {\footnotesize $k+2$} ] [.$\dots$ ] [.$k$ {\footnotesize $k+1$} {\footnotesize $k+2$} {\footnotesize $\dots$} {\footnotesize $2k-1$} ] ]
		\end{tikzpicture}
		\caption{\label{fig:branchtree} Decreasing the number of vertices in recursive calls of 
rule B2.}
	\end{center}
\end{figure}
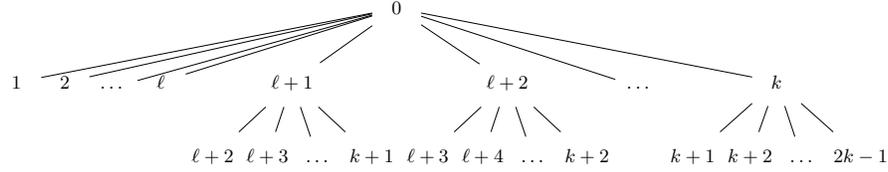
	We consider two cases. In the first case, $\ell=1$.
	Since $v_1$ is discarded in branches $2,\dots, k$, we have that the size of $e'$ is at most $k-1$ in each of these recursive calls, and the algorithm will either use branching rule B1 or branching rule B2 on a hyperedge of size at most $k-1$ in branches $2,\dots,k$.
	In the worst case, it uses branching rule B2 on a hyperedge of size $k-1$ in each of these branches, leading to the branching vector $(1,3,4,\dots,k+1,4,5,\dots,k+2,5,6,\dots,2k-1)$ whose recurrence \shortversion{solves to $\beta_k$:}%
	\begin{align*}
	T(n) = T(n-1) + \sum_{i=3}^k (i-2)\cdot T(n-i) + \sum_{i=k+1}^{2k-1}(2k-i)\cdot T(n-i)\shortversion{\enspace.}
	\end{align*}
	\longversion{solves to $\beta_k$. \par}%
	In the second case, $\ell\ge 2$, and we will show that this case is no worse than the first one.
	Since $v_1,\dots,v_\ell$ are discarded in branches $\ell+1,\dots, k$, the size of $e'$ is at most $k-\ell$ in each of these recursive calls, and in the worst case the algorithm will branch on a hyperedge of size $k-\ell$ in branches $\ell+1,\dots,k$, leading to the branching vector $(1,2,\dots,\ell,\ell+2,\ell+3,\dots,k+1,\ell+3,\ell+4,\dots,k+2,\dots,\allowbreak k+1,k+2,\dots,2k-\ell)$. See \autoref{fig:branchtree}. To see that this is no worse than the branching vector of the first case, follow each branch $i$ with $2\le i\le \ell$ by a $k$-way branching $(1,2,\dots,k)$, replacing the entry $i$ in the branching vector with $1+i,2+i,\dots,k+i$. Compared with the branching vector of the first case, the only difference in branches $i$, $2\le i\le \ell$, is that these have the additional entries $k+i$. But note that branch $\ell+1$ has entries $k+2,\dots,k+\ell$ in the first case but not in the second case. We conclude that the branching vector where
	entries $i$, $2\le i\le \ell$, are replaced by $1+i,2+i,\dots,k+i$ is a sub-vector of the one for $\ell=1$.
\qed
\end{proof}

Since this algorithm guarantees branching on hyperedges of size at most $k-1$ in certain cases, its running time outperforms the one in \cite{Gaspers05} for each $k\ge 3$.

\section{Lower bounds}

The graphs with a maximum number of maximal independent sets are the disjoint unions of triangles. They are hypergraphs of rank $2$ with $3^{n/3}$ minimal transversals.
We generalize this lower bound to hypergraphs with larger rank.

\begin{theorem}\label{thm:mhslb}
	For any two integers $k,n>0$, there is an $n$-vertex hypergraph of rank $k$ with
\longversion{
	\begin{align*}
	  \binom{2\cdot k -1}{k}^{\lfloor n/(2\cdot k-1) \rfloor}
	\end{align*}
}\shortversion{
 $\binom{2\cdot k -1}{k}^{\lfloor n/(2\cdot k-1) \rfloor}$
}
	minimal transversals.
\end{theorem}
\newcommand{\proofthmmhslb}{%
\begin{proof}
	Let $H_k=(V,E)$ be a hypergraph on $2\cdot k -1$ vertices where $E$ is the set of all subsets of $V$ of cardinality $k$.
	We claim that every subset of $k$ elements from $V$ is a minimal transversal of $H_k$. First, any transversal $X$ of $H_k$ contains at
	least $k$ vertices, otherwise $|V\setminus X|\ge k$ and there is least one hyperedge 
	that has an empty intersection with $X$. Second, any subset $X'$ of $V$ of cardinality $k$ is a hitting set for $H_k$. Indeed, as $|V\setminus X'|=k-1$,
	every hyperedge that does not intersect $X'$ has cardinality at most $k-1$, but $E$ contains no such hyperedge.
	This proves the claim that every subset of $k$ elements from $V$ is a minimal transversal of $H_k$. Thus, $H_k$ has
	\begin{align*}
	\binom{2\cdot k -1}{k}
	\end{align*}
	minimal transversals.
	The bound of the theorem is then achieved by a disjoint union of $\lfloor n/(2\cdot k-1) \rfloor$ copies of $H_k$, and
	$n-\lfloor n/(2\cdot k-1) \rfloor$ isolated vertices.
	\qed
\end{proof}
}
\longversion{\proofthmmhslb}

\longversion{
\section*{Acknowledgments}

We thank Fabrizio Grandoni for initial discussions on this research.

Serge Gaspers is the recipient of an ARC Future Fellowship (project number FT140100048) and acknowledges support under the ARC's Discovery Projects funding scheme (project number DP150101134).
NICTA is funded by the Australian Government through the Department of Communications and the Australian Research Council (ARC) through the ICT Centre of Excellence Program.
}

\shortversion{
 \bibliography{shortnames,literature}
 \bibliographystyle{serge-short}	
}
\longversion{
 \bibliography{longnames,literature}
 \bibliographystyle{plain}
}

\shortversion{
\newpage
\section*{Appendix: Proofs omitted in the main part}

\subsection*{Proof of \autoref{thm:ic}}
\proofthmic

\subsection*{Proof of \autoref{thm:mhslb}}
\proofthmmhslb
}

\end{document}